\documentclass[letterpaper]{article}

\usepackage{amsmath}
\usepackage{amssymb}
\usepackage{amsthm}
\usepackage[square]{natbib}
\usepackage{enumitem}

\newtheorem{theorem}{Theorem}

\newtheorem{lemma}[theorem]{Lemma}
\newtheorem{observation}[theorem]{Observation}
\newtheorem{proposition}[theorem]{Proposition}

\newtheorem{definition}{Definition}[section]

\newcommand{\sproof}{\noindent{\bf Proof.}\hspace*{1em}}

\def\eps{{\varepsilon}}

\def\literalqed{{\ \nolinebreak\hfill\mbox{\quad}}}

\long\def\symbolfootnote[#1]#2{\begingroup%
\def\thefootnote{\fnsymbol{footnote}}\footnote[#1]{#2}\endgroup} 

    \makeatletter
    \renewcommand\part{%
      \if@openright
        \cleardoublepage
      \else
        \clearpage
      \fi
      \thispagestyle{empty}%
      \if@twocolumn
        \onecolumn
        \@tempswatrue
      \else
        \@tempswafalse
      \fi
      \null\vfil
      \secdef\@part\@spart}
    \makeatother

\def\cite{\citep}

\def\eps{\epsilon}

\newcommand{\ord}[1]{\succ_{#1}}

\newcommand{\ceil}[1]{\left\lceil #1 \right\rceil}

\def\({\left(}
\def\){\right)}

\newcommand{\xqed}{\mbox{\raggedright $\Diamond$}}
\def\xqedhere{\hfill\xqed}

\newcommand{\step}[1]{\stackrel{{\scriptscriptstyle{#1}}}{\rightarrow}}

\renewcommand{\vec}{\mathbf}

\begin{document}
\title{Random Tie-breaking with Stochastic Dominance}
\author{Reshef Meir\\ \texttt{reshefm@ie.technion.ac.il}}
\maketitle

\begin{abstract}
Consider Plurality with random tie-breaking.
This paper uses standard axiomatic extensions of preferences over elements to preferences over sets (Kelly, Gardenfors, Responsiveness) to characterize all better-replies of a voter under stochastic dominance. 
\end{abstract}

\section{Introduction}
Suppose that a decision maker has to select a subset of alternatives $W\subseteq C$, $|C|=m$. The agent has linear  preferences order   $Q$ over all alternatives $c\in C$.
Denote by $C^{(j)}$ the $j$ most preferred items in $C$ according to $Q$. Lottery $p$ \emph{stochastically dominates} $p'$ according to preference $Q$ if for every $j\leq m$, $Pr_{w\sim p}(w\in C^{(j)}) > Pr_{w\sim p'}(w\in C^{(j)})$.

Here we assume that given a set $W$ of ``possibly winning outcomes'', the actual outcome $w\in W$ is selected from $W$ uniformly at random. Without further information or restrictions on the agent's preferences,  $Q$ can be extended to preferences $\hat Q$ over $2^A$ in various ways. In particular, $Q$ induces a partial preference order $\hat Q$ once we enforce \emph{stochastic dominance}: Each set $W$ determines a lottery $p_W$ over outcomes in $C$ (by our assumption, a uniform lottery over $W$). We say that $X$ \emph{stochastically dominates} $Y$ is $p_X$ stochastically dominates $p_Y$. Intuitively, it means that the agent should prefer $X$ over $Y$ if she believes tie-breaking is going to be selected uniformly at random from the set, regardless of anything else.

To see that SD is only a partial relation, consider the preference $a \succ b \succ c$, and the sets $X=\{b\}, Y=\{a,c\}$.  Whether the agent prefers $X$ or $Y$ may depend on her cardinal utilities. E.g. for utilities $(4,2,1)$ we have $Y \succ X$, whereas for utilities $(4,3,1)$ we have $X \succ Y$. 
More generally, it is known that $X$ stochastically dominates $Y$ if and only if $X \succ Y$ for any cardinal utility scale $u$ consistent with $Q$. 

\subsection{Axioms}
Let $\hat Q$ be the partial order over $2^C$ derived from $Q$ and SD. We would like to find a natural axiomatic characterization of $\hat Q$, i.e. one that uses familiar axioms rather than lotteries and cardinal utilities. 

Here are three axioms that have been suggested in the literature for extending preferences over elements to preferences over subsets.

\begin{tabular}{|l|l|l|}
\hline
Axiom &  Name and reference & Definition \\
\hline
\hline
\textbf{K1}  & Kelly~\cite{kelly1977strategy} &$(\forall a \in X, b\in Y, a \succ  b)\ \Rightarrow\ X \succ  Y$ \\
\textbf{K2}  & ``~~~  `` &     $(\forall a \in X, b\in Y, a \succeq  b)\ \Rightarrow\ X \succeq  Y$ \\
\hline
\textbf{G}	& G\"ardenfors~\cite{gardenfors1976manipulation} & $(\forall b \in X, a \ord i b)\ \Rightarrow\  \{a\} \succ  (\{a\} \cup X) \succ  X$ \\
\hline
\textbf{R} & Responsiveness~\cite{roth1985college} &  $a \ord i b\ \iff\ \forall X\subseteq C\setminus\{a,b\},\ (\{a\} \cup X) \succ  (\{b\} \cup X)$\\
\hline
\end{tabular}

\subsection{Contribution}
If $X$ stochastically dominates $Y$, then this cannot violate any of the Axioms K+G+R. Yet, it is possible that $X$ SD $Y$ but this does not follow from the axioms. 
In the paper we show two results.
\begin{enumerate}
	\item Suppose the winner set $W$ is the outcome of the Plurality rule.
	The axioms~K+G+R characterize $\hat Q$ on all pairs of outcomes $X,Y$ such that a single voter can change the outcome from $X$ to $Y$. That is, the axioms characterize all better-replies in the game defined by the Plurality voting rule over candidates $C$ with uniform random tie breaking and stochastic dominance.
		\item We introduce another axiom called \emph{monotone duplication} (MD) s.t. axioms K+MD+R characterize $\hat Q$ for any pair of outcomes. 
\end{enumerate}

\section{Characterization of Better-Replies in Plurality}

\begin{definition}\label{def:matching} 
Suppose that $X,Y\subseteq C$, $k=|X|\leq |Y|=K$. Sort $X,Y$ in increasing order by $Q$. Let $r_j = \ceil{\frac{j}{k}K}$. 
Partition $Y$ into sets $Y_1,\ldots,Y_k$ s.t. for $j<K$,  $Y_j=\{y_{r_{j-1}+1},\ldots, y_{r_j}\}$ (e.g., if $k=3,K=7$, then $Y$ is partitioned into $Y_1=\{y_1,y_2,y_3\}, Y_2=\{y_4,y_5\}, Y_3= \{y_6,y_7\}$).

$X$ \emph{match-dominates} $Y$ according to $Q$ if:
\begin{itemize}
	\item (I) $\forall j\leq k \forall y\in Y_j$, $x_j \succeq y$; and
	\item either (IIa) at least one relation is strict, or (IIb) $K \mod k\neq 0$.
\end{itemize}
If $|X|>|Y|$, then $X$ \emph{match-dominates} $Y$ if $Y$ match-dominates $X$ according to the reverse of $Q$. 
\end{definition} 
Intuitively, match-domination means that for any $q\in[0,1]$, there is a fraction $q$ of the set $X$ that  dominates a fraction of $1-q$ from the set $Y$: at least one $x\in X$ dominates all of $Y$, at least 20\% of $X$ dominate at least 80\% of $Y$, and so on. 

\begin{lemma}\label{lemma:SD}
Let $\vec a,\vec a'$ be two profiles that differ by a single vote, and define $X=f(\vec a),Y=f(\vec a')$.\footnote{Without some restriction on $X,Y$, the lemma is incorrect. E.g. if $x_1\succ y_1 \succ y_2 \succ x_2 \succ y_3 \succ y_4$, then $X$ stochastically dominates $Y$ but there is no way to derive $X\succ Y$ from the axioms K+G+R.}
 
The following conditions are equivalent for any strict order $Q$ over $C$:
\begin{enumerate}
	\item $X$ stochastically dominates $Y$ under preferences $Q$ and uniform lottery.
	\item The relation $X\succ Y$ is entailed by $Q$, Axioms~K+G+R, and transitivity.
	\item  $u(X)>u(Y)$ for every $u$ that is consistent with $Q$.
	\item  $X$ match-dominates $Y$ according to $Q$. 
\end{enumerate} 	
\end{lemma}
\begin{proof}
The equivalence of (1) and (3) is immediate for any sets $X,Y$, and used e.g. in \cite{reyhani2012best}.

(2) $\Rightarrow$ (3). If $X\succ Y$ follows from the axioms, then there is a sequence of sets $X=X_0\succ X_1 \succ \cdots \succ X_k =Y$ such that each $X_j \succ X_{j+1}$ follows from a single axiom K,G, or R. Thus it is sufficient to show for $X\succ Y$ that follows from a single axiom. 

If $X\succ Y$ follows from Axiom~R, then $X=\{a\}\cup W,Y=\{b\}\cup W$ for some $W\subseteq C\setminus\{a,b\}$ and $a\succ b$. Thus
$$u(X) = u(\{a\} \cup W) = \frac{1}{|W|+1|}\(u(a) + \sum_{c\in W}u(c)\) > \frac{1}{|W|+1|}\(u(b) + \sum_{c\in W}u(c)\) = u(\{b\}\cup W)=u(Y).$$

If $X\succ Y$ follows from Axiom~G, then either $X=Y\cup \{a\}$ and $a\succ b$ for all $b\in Y$, or $X=\{x\}$ and $Y=\{x\}\cup W$ where $x\succ w$ for all $w\in W$. For the first case
\begin{align*}
u(X) &=\frac{1}{|Y|+1|}u(a) + \frac{1}{|Y|+1}\sum_{y\in Y}u(y) =\frac{1}{|Y|+1}\frac{1}{|Y|}\sum_{y\in Y}u(a) +   \frac{1}{|Y|+1}\sum_{y\in Y}u(y)\\
&> \frac{1}{|Y|+1|}\frac{1}{|Y|}\sum_{y\in Y}u(y) +   \frac{1}{|Y|+1|}\sum_{y\in Y}u(y)\\
&=\(1+\frac{1}{|Y|}\) \frac{1}{|Y|+1}\sum_{y\in Y}u(y) =  \frac{1}{|Y|}\sum_{y\in Y}u(y) = u(Y).
\end{align*}
For the second case, 
$$u(X) = u(x) = \frac{1}{|Y|}\sum_{y\in Y}u(x) = \frac{1}{|Y|}\(u(x) +\sum_{w\in W}u(x)\) >  \frac{1}{|Y|}\(u(x) +\sum_{w\in W}u(w)\) = u(Y).$$
 
If $X\succ Y$ follows from Axiom~K, then $u(x)>u(y)$ for any $x\in X,y\in Y$ which is a trivial case.


(3) $\Rightarrow$ (4). 
Suppose that $u(X)>u(Y)$ for all $u$. Suppose first $|X|\leq |Y|$. If $|X|$ does not match-dominate $Y$ then either (I) there  is an element $x_{j'}$ that is less preferred than some element $y'\in Y_{j'}$; or (II) for all $j$ and all $y\in Y_j$, $x_j=_Q y$ and $|Y_j|=\frac{K}{k}=q$ for all $j$. We will derive a contradiction to (3) in either case. In the latter case, we have $u(x_j)=u(Y_j)$ for all $j$ and thus 
$$u(Y)=\frac{1}{K}\(\sum_{j\leq k}|Y_j|u(Y_j) \)=\frac{\sum_{j\leq k}q u(x_j) }{K} = \frac{\sum_{j\leq k}q u(x_j)}{kq} =  u(X),$$
In contradiction to (3).

Thus we are left with case (I). That is, there are $j' \leq k$ and $y'\in Y_{j'}$ s.t. $x_{j'} \prec y'$.  We define the (possibly empty) set $X'\subseteq X$ as all elements $\{x: x\succ x_{j'}\}$. We define $Y'\subseteq Y$ as $\{y: y\succeq y'\}$. By construction, for any $j>j'$, $Y_j\subseteq Y'$. Thus
 $$|Y'|\geq 1+\sum_{j=j'+1}^k|Y_j|=1+\sum_{j=j'+1}^k(r_j-r_{j-1})=(K-r_{j'})+1 = (K-\ceil{\frac{j'}{k}K})+1> K-\frac{j'}{k}K=K(1-\frac{j'}{k}),$$
whereas $|X'|\leq k-j'$. 
We define $u$ as follows: $u(x)=1,u(y)=1$ for all $x\in X',y\in Y$, and $u(z)=0$ for all other elements. Note that $X',Y'$ contain the top elements of $X,Y$, respectively. In addition, $y'$ is the minimal element in $Y'$ and by transitivity $y'\succ x$ for all $x\in X\setminus X'$. Thus $u$ is consistent with $Q$.\footnote{If we want $u$ to respect the strict order $Q$, we can vary the cardinal preferences within each set $X',Y',X\setminus X',Y\setminus Y'$ s.t. the differences within each set are less than $\eps$. For sufficiently small $\eps$ (say, $\eps<\frac{1}{k^2}$) the proof goes through.} 

We argue that $u(Y)>u(X)$ in contradiction to (3). Indeed, $u(X) = \frac{|X'|}{|X|}\leq \frac{k-j'}{k}=1-\frac{j'}{k}$. 
$$u(Y) = \frac{|Y'|}{|Y|} > \frac{(1-\frac{j'}{k})K}{K} = 1-\frac{j'}{k} = \frac{k-j'}{k} \geq \frac{|X'|}{|X|}= u(X),$$
so we get a contradiction to (3) again.  
Thus $X$ matching-dominate $Y$.

(4) $\Rightarrow$ (2).
This is the only part of the proof where we use the profiles from which $X,Y$ are obtained.
When a single voter moves, either the winner set changes by a single candidate (added, removed, or swapped), or $X$ is a single candidate, or $Y$ is a single candidate. We prove case by case.
\begin{itemize}
\item The case where $|X|=|Y|=1$ is immediate. 
	\item Suppose $|X|=1$ (i.e. $X=\{x\}$) and $|Y|=K>1$. Then $X$ match-dominates $Y$ means that $x \succeq y$ for all $y\in Y$, with at least one relation being strict, w.l.o.g. $y_K$ (least preferred in $Y$). Then $X \succeq \{y_1,\ldots,y_{K-1}\} \succ Y$, where the first transition is by Axiom~K2 and the second is by Axiom~G.
	\item The case of $|Y|=1$ is symmetric.
	\item Suppose $|X|=|Y|=k$. Then  $X$ match-dominates $Y$ means that $x \succeq y $ for all $i$. For all $t\in\{0,1,\ldots,k\}$, let $X^t=\{x_1,\ldots,x_t,y_{t+1},\ldots,y_k\}$. Then $X^{t-1} = X^{t}$ if $x_t=y_t$, and $X^{t-1} \succ X^{t}$ otherwise from Axiom~R. In addition, $X=X^0,Y=X^k$ thus $X\succ Y$ from transitivity. 
	\item Suppose $|X|=k,|Y|=k+1$. Then $X$ match-dominates $Y$ means that $|Y_1|=\ceil{\frac{k+1}{k}}=2$, and all other sets $Y_j$ are singletons $Y_j=y_j$. Consider the set $Y'$ that includes the top $k$ elements of $Y$. Since $x_1$ is (weakly) preferred to both candidates in $Y_1$, $Y'$ is match-dominated by $X$. By the previous bullet $X \succeq Y'$ follows from Axiom~R and transitivity. Finally, $Y'\succ Y=Y'\cup\{\min Y\}$ by Axiom~G.
\end{itemize}
\end{proof}

The following is an immediate corollary:

\begin{proposition}\label{th:SD_axioms}
 A step $\vec a\step i \vec a'$ is a better-response under random tie-breaking and stochastic dominance, if and only if $f(\vec a') \succ  f(\vec a)$ is entailed by $Q $, the Axioms~K+G+R, and transitivity.
\end{proposition}

\section{A Full Characterization of Stochastic Dominance Preferences}
Lemma~\ref{lemma:SD} provides an axiomatic characterization for any pair of subsets that are the result of a single voter move (under Plurality). What if we want to characterize all pairwise relations? 
To that end, we need another axiom. In addition, the set-extension of $Q$ applies to multisets and not just to sets. 

\begin{definition}
Consider a preference order $Q$ over a set $C$. 
A partial extension $\hat Q$ to multisets (that are subsets of $C$) respects \emph{Monotone Duplication} (\textbf{MD}) if the following holds for any $X=\{x_1,\ldots,x_k\}$ in non-decreasing order according to $Q$, and $Y$ is a multiset where each $x_j$ appears $h_j\geq 0$ times:
\begin{enumerate}
	\item If $(h_j)_{j\leq k}$ is non-decreasing then $Y \succeq_{\hat Q} X$;
	\item If $(h_j)_{j\leq k}$ is non-increasing then $Y \preceq_{\hat Q} X$;
	\item  If $(h_j)_{j\leq k}$ are not all equal, then $Y \neq_{\hat Q} X$ (preference is strict). 
\end{enumerate}
\end{definition}

Since we now allow multisets, we require all Axioms to apply for weak preferences as well. 

\begin{tabular}{|l|l|}
\hline
Axiom &   Definition \\
\hline
\hline
\textbf{K1}   &$(\forall a \in X, b\in Y, a \succ  b)\ \Rightarrow\ X \succ Y$ \\
\textbf{K2}  &     $(\forall a \in X, b\in Y, a \succeq  b)\ \Rightarrow\ X \succeq  Y$ \\
\hline
\textbf{G1}	& $(\forall b \in X, a \ord i b)\ \Rightarrow\  \{a\} \succ  (\{a\} \cup X) \succ  X$ \\
 \textbf{G2} & $(\forall b \in X, a \succeq  b)\ \Rightarrow\  \{a\} \succeq  (\{a\} \cup X) \succeq  X$ \\
\hline
\textbf{R1} &  $a \succ  b\ \iff\ \forall X\subseteq C\setminus\{a,b\},\ (\{a\} \cup X) \succ  (\{b\} \cup X)$\\
\textbf{R2} &  $a \succeq  b\ \iff\ \forall X\subseteq C\setminus\{a,b\},\ (\{a\} \cup X) \succeq  (\{b\} \cup X)$\\

\hline
\end{tabular}

\begin{observation}
Axiom~MD entails Axiom~G.
\end{observation}
To see why, note that given a set $A$ and $b$ s.t. $b\succ a$ for all $a\in A$, we can define $X=A\cup\{b\},Y_1=\{b\}$, then $Y_1$ is obtained from $X$ by duplication with $h_1=\cdots=c_{|A|-1}=0,h_{|A|}=1$. Thus by Axiom~MD1+MD3 we get $\{b\} \succ \{b\} \cup A$. Similarly, for $Y_2= A$ we get for $h_1=\cdots=c_{|A|-1}=1,h_{|A|}=0$ and Axioms~D2+D3 that $\{b\} \cup A \succ A$.  

\begin{theorem}
\label{th:SD_all}
Let $X,Y \subseteq C$.
The following conditions are equivalent for any strict order $Q$ over $C$:
\begin{enumerate}
	\item $X$ stochastically dominates $Y$ under preferences $Q$ and uniform lottery.
	\item The relation $X\succ Y$ is entailed by $Q$, Axioms~K+MD+R, and transitivity.
	\item  $u(X)>u(Y)$ for every $u$ that is consistent with $Q$.
	\item  $X$ match-dominates $Y$ according to $Q$. 
\end{enumerate} 	

\end{theorem}
\begin{proof}
Lemma~\ref{lemma:SD} already shows (1)$\iff$(3) and (2)$\Rightarrow$(3)$\Rightarrow$(4) for Axioms~K+G+R, without any further condition. Since Axiom~MD is stronger than G these entailments still hold for Axioms~K+MD+R.  It remains to show that (4)$\Rightarrow$(2) under these axioms.  

We prove for the case where $|X|\leq |Y|$. We start with $Y$ and generate a multisets $Y'$ by for each $j\leq k$ replacing each $y\in Y_j$ with $|Y_j|$ copies of $x_j$. Since $x_j \succeq y$ for each such replacement, we get $Y' \succeq Y$ by (weak) Axiom~R and transitivity.

Then by Axiom~MD2  we have that $X\succeq Y'$, since for every $j<j'\leq k$, $Y'$ contains (weakly) more copies of $x_j$ than of $x_{j'}$. Therefore $X\succeq Y'\succeq Y$. 

To show that the preference is strict, we consider the two cases of Def.~\ref{def:matching}: either (IIa) at least one relation $x_j\succ y$ is strict, or (IIb) $K \mod k\neq 0$.  In case (IIa), we get that $Y'\succ Y$ by strict Axiom~R. In Case (IIb), we get that the sizes $h_j$ in $Y'$ are not all the same, and thus by Axioms~MD2+MD3 $X\succ Y'$.

In either case, $X\succ Y$ so we are done.
\end{proof}

\bibliographystyle{named}
\bibliography{abbshort,plurality.aaai,ultimate} 
\end{document}